\documentclass[conference]{IEEEtran}
\IEEEoverridecommandlockouts
\usepackage{cite}
\usepackage{amsmath,amssymb,amsfonts}
\usepackage{algorithmic}
\usepackage{graphicx}
\usepackage{textcomp}
\usepackage{xcolor}
\usepackage{stfloats}
\usepackage{bm}
\usepackage{epstopdf}
\usepackage{multirow}
\usepackage{enumitem}
\usepackage{stfloats}
\usepackage{amsthm}
\usepackage[ruled,linesnumbered]{algorithm2e}
\usepackage{amsthm}
\theoremstyle{remark}
\newtheorem{remark}{\textbf{Remark}}
\newtheorem{theorem}{Theorem}

\def\BibTeX{{\rm B\kern-.05em{\sc i\kern-.025em b}\kern-.08em
    T\kern-.1667em\lower.7ex\hbox{E}\kern-.125emX}}

\expandafter\def\expandafter\normalsize\expandafter{%
	\normalsize%
	\setlength\abovedisplayskip{1pt}%
	\setlength\belowdisplayskip{5pt}%
}

\begin{document}
\title{Near-Field Channel Estimation for mmWave/THz Communications with Extremely Large-Scale UPAs }
\author{\IEEEauthorblockN{Yiming Chen, Hongwei Wang, Lingxiang Li, and Zhi Chen}
\thanks{Corresponding author: Hongwei Wang}
\thanks{This work was supported in part by the National Key R\&D Program of China under Grant 2024YFE0200400, the Mobile Information Networks National Science and Technology Major Project under Grants No. 2025ZD1303000 and No. 2025ZD1302000, the Fundamental Research Funds for the Central University under Grant ZYGX2022J029, and the Natural Science Foundation of Sichuan Province under Grant 2025ZNSFSC0514.}
\IEEEauthorblockA{\textit{National Key Laboratory of Wireless Communications} \\
	\textit{University of Electronic Science and Technology of China}, Chengdu, P.R. China\\
	202521220228@std.uestc.edu.cn, \{hongwei\_wang, lingxiang.li, chenzhi\}@uestc.edu.cn     }
	}
\maketitle

\begin{abstract}
Extremely large antenna arrays (ELAAs) are widely adopted in mmWave/THz communications to compensate for the severe path loss, wherein the channel estimation remains a significant challenge since the Rayleigh distance of ELAAs stretches to tens or even hundreds of meters and the near-field channel model should be considered. Existing polar-domain based methods and block-sparse based methods are originally devised for Uniform Linear Arrays (ULAs) near-field channel estimation. The polar-domain based method can be applied to Uniform Planar Arrays (UPAs), but it behaves plain since it ignores the specific sparsity structure of the UPA near-field channels. Meanwhile, the block-sparse based method cannot be extended to the UPA scenarios directly. To address these issues, we first reformulate the original UPA near-field channel as an outer product of two ULA near-field channels and we construct a modified two-dimensional DFT (2D-DFT) dictionary for it. With the proposed dictionary, we further prove that the UPA near-field channel admits a 2D block-sparse structure. Leveraging this specific sparse structure, we solve the channel estimation problem with the 2D Pattern-Coupled Sparse Bayesian Learning (2D-PCSBL) algorithm. Simulation results show that the proposed approach outperforms conventional existing methods while maintaining a comparable computational complexity.
\end{abstract}

\vspace{-2mm}
\section{Introduction}
ELAAs have emerged as a key enabling technology for next-generation communications, which can achieve ultra-high array gains and offer an effective solution for the severe path loss in mmWave/THz communications~\cite{ELAAHybridBeamforming}. 
Unfortunately, accurate channel state information (CSI) is essential for fully unlocking the potential of ELAAs, while acquiring such CSI poses significant practical challenges. As their Rayleigh distance stretches to tens or even hundreds of meters~\cite{cui2022near,ye2024extremely}, the spherical wavefront assumption should be employed to accurately capture the propagation characteristics, thereby yielding a near-field channel determined by both the distance and the angular parameters. In such a case, conventional far-field channel estimation methods such as orthogonal matching pursuit (OMP)~\cite{lee2016channelOMP}, sparse Bayesian learning (SBL)~\cite{EMSBL}, and atomic norm minimization (ANM)~\cite{ANMRIS} are no longer applicable.

 A near-field sparse representation was first proposed in~\cite{LS}, where the dictionary is constructed by jointly sampling the two-dimensional angular-range domain, also referred to as the polar-domain method. It effectively mitigates the energy spread issue induced by the conventional angular-domain transformations. Nevertheless, the number of atoms in this representation far exceeds the channel dimension, which significantly increases the computational complexity. Moreover, since the dictionary used is non-unitary and generally possesses high coherence, it yields a poor restricted isometry property (RIP), thereby degrading the channel estimation performance. To handle these issues, the works in~\cite{wang2024nearfarfieldchannelestimationterahertz, WangGlob} introduce a modified DFT matrix as the dictionary and represent the near-field channel in a block-sparse pattern, based on which methods like block-OMP (BOMP)~\cite{Blocksparsesignal} and Pattern-Coupled Sparse Bayesian Learning (PCSBL)~\cite{1D} can be employed to estimate the near-field channel with reduced pilot overhead.

 All the aforementioned methods are originally designed for ULAs. Although the polar-domain based method can be extended to the UPA scenario, it performs non-competitively in channel estimation accuracy. By contrast, the block-sparse based method, though offering promising estimation performance with low pilot overhead in ULA systems, cannot be directly extended to UPA scenarios.
 In this paper, we aim to exploit the coupling feature of the sparsity across the horizontal and vertical dimensions of the UPA near-field channels to improve channel estimation performance and reduce the required pilot overhead. To the best of our knowledge, this is the first attempt to exploit this coupling feature of the UPA near-field channels.

Our main contributions are summarized as follows:
\begin{itemize}
	\item[(1)] We reformulate the original UPA near-field channel as an outer product of two ULA near-field channels, upon which a modified 2D-DFT matrix is designed to serve as the near-field channel dictionary.
	
	\item[(2)] We prove that under the proposed 2D-DFT dictionary, the UPA near-field channel exhibits a 2D block-sparse structure, based on which we reformulate the UPA channel estimation task as a 2D block-sparse signal recovery problem.
	
	\item[(3)] We solve the reformulated problem using the 2D-PCSBL algorithm~\cite{twodimention}. Simulation results demonstrate that, compared with the existing methods, the proposed approach attains superior channel estimation accuracy and reduced pilot overhead, while maintaining a comparable computational complexity. 
\end{itemize}

\section{Signal Model and Problem Formulation}
\vspace{-2mm}
\subsection{Scenario and Signal Model}
Consider a downlink transmission scenario where a wireless communication system operating at $f_c$ serves multiple single-antenna users. The transmitter is equipped with an ELAA configured as a UPA comprising $N=N_x\times N_y$ antennas, with $N_x$ and $N_y$ being, respectively, the number of antennas along the horizontal and vertical axes. To reduce hardware complexity, the transmitter adopts a hybrid beamforming architecture with $R\ll N$ radio frequency (RF) chains. The carrier wavelength is $\lambda_c=\frac{c}{f_c}$, where $c$ is the speed of light. The inter-element spacing is $d_x=d_y=d=\lambda_c/2$ along both axes. Each user estimates its channel from the pilot signals broadcast by the transmitter during the downlink transmissions.

In this paper, we have the following three assumptions. Firstly,  the users fall within the transmitter's near-field region, and a spherical wavefront assumption is applied. Secondly, the sparse scattering characteristic of mmWave/THz signals leaves only a few dominant paths between the transmitter and the user. Thirdly, due to the limit of space, in this paper, we only focus on the narrowband case and will extend the following proposed framework to broadband near-field channels.

Under these assumptions, the channel between the transmitter and the user, denoted by $\mathbf H \in \mathbb{C}^{N_x \times N_y}$, is modeled as

\begin{align} \label{H_ori}
\mathbf{H}  =\sum_{l=0}^{L-1}g_le^{-j2\pi f_c\tau_l}\mathbf{A}(r_l,\theta_l,\phi_l)=\sum_{l=0}^{L-1}\tilde{g}_l\mathbf{A}(r_l,\theta_l,\phi_l),
\end{align}
where $L$ is the number of propagation paths; $\tilde{g}_l \triangleq g_l e^{-j 2 \pi f_c \tau_l}$, with $g_l$  and $\tau_l$ denoting the complex channel gain and the propagation delay from the transmitter reference antenna to the target of the $l$-th path, respectively; $(r_l,\theta_l,\phi_l)$ denote the distance from the reference antenna at the transmitter to the target, and the elevation and azimuth angles associated with the $l$-th propagation path; and $\mathbf{A}(r, \theta, \phi) \in \mathbb{C}^{N_x \times N_y}$ denotes the near-field steering matrix of the UPA. Specifically, $\mathbf{A}(r, \theta, \phi)$ is expressed as

\begin{align}\label{eq:steering_vector}
&\mathbf{A}(r, \theta, \phi) \nonumber\\
&\  \triangleq \frac{1}{\sqrt{N}}
\begin{bmatrix}
e^{-\frac{j 2\pi}{\lambda_c} \left( r^{(0,0)} - r \right)} & \cdots & e^{-\frac{j 2\pi}{\lambda_c} \left( r^{(N_y-1,0)} - r \right)} \\
\vdots & \ddots & \vdots \\
e^{-\frac{j 2\pi}{\lambda_c} \left( r^{(0,N_x-1)} - r \right)} & \cdots & e^{-\frac{j 2\pi}{\lambda_c} \left( r^{(N_x-1,N_y-1)} - r \right)}
\end{bmatrix},
\end{align}
where $n_x \in \{0,1,\dots,N_x-1\}$ and $n_y \in \{0,1,\dots,N_y-1\}$. $r^{(n_x,n_y)}$ denotes the distance from the antenna at position $(n_x,n_y)$ in the UPA to the target user. $r=r^{(0,0)}$ represents the distance from the reference antenna $(0,0)$ to the user.

The transmitted signal matrix of the UPA array 
at the $t$-th sampling instant, denoted by 
$\mathbf{X}(t) \in \mathbb{C}^{N_x \times N_y}$, can be expressed as
\begin{align}
	\mathbf{X}(t) = \mathbf{B}(t) \odot \mathbf{S}(t),
\end{align}
where $\odot$ represents the Hadamard product. Here, $\mathbf{B}(t)$ denotes the hybrid beamforming matrix satisfying the constant modulus constraint $|\mathbf{B}(i,j)|=1/\sqrt{N}$. $\mathbf{S}(t)$ represents the pilot signal matrix. Define $T$ as the sampling number, without loss of generality, we assume that $\mathbf{S}(t) = \mathbf{1}_{N_x \times N_y}$ for $t=1,2,..., T$, which leads to $\mathbf{X}(t) = \mathbf{B}(t)$. 
Then, the received signal at the $t$-th sampling instant can be written as
\begin{align} \label{receivedyt}
	y(t) = \sum_{i=1}^{N_x}\sum_{j=1}^{N_y} [\mathbf{H}]_{ij}[\mathbf{X}(t)]_{ij}
	= \mathbf{h}^{T}\mathbf{f}(t),
\end{align}
where $\mathbf{h} \triangleq \mathrm{vec}(\mathbf{H})$ and $\mathbf{f}(t) = \mathrm{vec}(\mathbf{B}(t))$. 
By stacking the $y(t)$ over $T$, 
we obtain the received signal vector as
\begin{align} \label{receivedmodelY}
	\mathbf{y} = \mathbf{Fh} + \mathbf{n},
\end{align}
where $\mathbf{y} \triangleq \left[ y(1) \ \cdots y(T) \right]^{T} \in \mathbb{C}^{T}$, $\mathbf{F} \triangleq \left[ \mathbf{f}(1) \ \cdots \ \mathbf{f}(T) \right]^{T} \in \mathbb{C}^{T \times N}$ and $\mathbf{n} \sim \mathcal{CN}(\bm{0}, \sigma_n^2 \bm{I}_T)$ represents additive white Gaussian noise (AWGN).

\subsection{Problem Formulation}
We aim to estimate $\mathbf{h}$ from the received signal $\mathbf{y}$, as follows,
\begin{align}
 \hat{\mathbf{h}} = 
 \arg\min_{\mathbf{h}} \; \|\mathbf{y}-\mathbf{F}\mathbf{h}\|.
\end{align}
  
  Since this problem involves reconstructing a high-dimensional signal $\mathbf{h}$ from a limited observation $\mathbf{y}$, employing the compressive sensing framework offers a more pilot-efficient approach. Unfortunately, the co-existing polar-domain based methods cannot be efficiently applied to UPAs, and the block-sparse based methods remain unexplored in such scenarios.
  
In this paper, we aim to exploit the coupling feature of the sparsity across the horizontal and vertical dimensions of the UPA near-field channels to improve channel estimation performance and reduce required pilot overhead.

\section{Proposed 2D Block-Sparse Aware Channel Estimation Method}
\begin{figure*}[!hb]
	\centering
	\vspace*{2pt}
	\hrulefill
	\vspace*{2pt}
	\setlength{\jot}{-2pt} 
	\begin{align}
		r^{(n_x,n_y)} &=\sqrt{(r\sin\theta\cos\phi)^2+(r\cos\theta-n_xd)^2+(r\sin\theta\sin\phi-n_yd)^2} \notag \\
		&\overset{(a)}{\operatorname*{\operatorname*{\approx}}}r - n_x d \cos\theta - n_y d \sin\theta \sin\phi + \frac{n_x^2 d^2}{2 r} (1 - \cos^2\theta) + \frac{n_y^2 d^2}{2 r} (1 - \sin^2\theta \sin^2\phi) - \frac{n_x n_y d^2 \cos\theta \sin\theta \sin\phi}{r} \notag \\
		&\overset{(b)}{\operatorname*{\operatorname*{\approx}}} r - d \left[ n_x\zeta_a + n_y\zeta_e \right]
		+ \frac{d^2}{2r} \left[ n_x^2 (1 - \zeta_a^2) + n_y^2 (1 - \zeta_e^2)  \right]
		\label{app_distant}
	\end{align}
\end{figure*}
In this section, we first represent the near-field UPA channel $\mathbf{H}$ as the outer product of two near-field ULA channels, based on which we construct 
a modified 2D-DFT matrix, serving as the near-field channel dictionary. Then, by leveraging the 2D block-sparse structure, we solve the channel estimation problem using the 2D-PCSBL algorithm. 

Finally, we analyze the computational complexity of the proposed algorithm and compare it with the UPA-extended polar-domain method and other block-sparse based methods.

\subsection{Sparse Property of the UPA Near-Field Channel}
The steering matrix of the UPA is mainly based on $r^{(n_x, n_y)}$ as shown in \eqref{eq:steering_vector}. Therefore, we firstly investigate the exact expression of $r^{(n_x, n_y)}$ which is given by~\eqref{app_distant} at the bottom of next page, where $\zeta_a \triangleq  \cos\theta$ and $\zeta_e \triangleq \sin\theta \sin\phi$. In (a) of~\eqref{app_distant} we employ the approximation $\sqrt{1+x^2}\approx 1 + 0.5x- 0.125x^2$ which holds true for small $x$, and in (b) we neglect the bilinear quadratic term. Such an approximation holds true in near-field regions, and the detailed explanation is provided in~\cite{Muitipleaccess}.

According to the approximation in~\eqref{app_distant}, the near-field steering matrix $\mathbf{A}(r,\theta,\phi)$ of a UPA can be decomposed as
\begin{align}
	\mathbf{A}(r,\theta,\phi) &\approx \left( \bar{\mathbf{a}}(\zeta_a) \odot \bar{\mathbf{b}}(\zeta_a, r) \right)\left( \bar{\mathbf{a}}(\zeta_e) \odot \bar{\mathbf{b}}(\zeta_e, r) \right)^{T} \notag \\
	&= \mathbf{a} (\zeta_a, r,N_x) \, \mathbf{a} (\zeta_e, r,N_y)^{T}.
	\label{app_steering}
\end{align}
Here, $\bar{\mathbf{a}}(\zeta) \in \mathbb{C}^{N \times 1}$, with $\zeta \in \{\zeta_a, \zeta_e\}$, 
and its $n$th element is defined as 
$[\bar{\mathbf{a}}(\zeta)]_n = e^{j \frac{2\pi}{\lambda_c} n d \zeta}$. 
Similarly, $\bar{\mathbf{b}}(\zeta, r) \in \mathbb{C}^{N \times 1}$ has entries given by 
$[\bar{\mathbf{b}}(\zeta, r)]_n = \frac{1}{\sqrt{N}} e^{-j \frac{2\pi}{\lambda_c} \frac{n^2 d^2}{2r}(1-\zeta^2)}$. 
In addition, $\mathbf{a}(\zeta, r, N) \in \mathbb{C}^{N \times 1}$ can be seen as the near-field steering vector of a ULA with its functional form remaining unchanged except for the parameter $\zeta$.
The $n$-th element of $ \mathbf{a}(\zeta,r,N) $ is given by

\begin{align}
	[\mathbf a(\zeta,r,N)]_{n} 
	\triangleq
	\frac{1}{\sqrt{N}}
	\exp\!\Big(
	-\,j\frac{2\pi}{\lambda_c}
	\big(-n d\zeta + \tfrac{n^2d^2}{2r}(1-\zeta^2)\big)
	\Big).
\end{align}

Substituting  \eqref{app_steering} into \eqref{H_ori}, the channel matrix $\mathbf{H}$ can be rewritten as
\begin{align}
	\mathbf{H}=  \sum_{l=0}^{L-1} \tilde{g}_l \, \mathbf{a}(\zeta_{a,l}, r_l,N_x) \, \mathbf{a}(\zeta_{e,l}, r_l,N_y)^{T},
	\label{eq:H_final}
\end{align}

From \eqref{eq:H_final}, one can see that we rewrite the channel matrix $\mathbf{H}$ as an outer product of two ULA near-field steering vectors.

According to our previous work~\cite{wang2024nearfarfieldchannelestimationterahertz}, a ULA near-field steering vector can be represented in block-sparse form on a specified modified DFT basis. Therefore, we have
\begin{align}
    \mathbf{a}(\zeta_{a,l}, r_l,N_x)&=\tilde{\mathbf{D}}_{x}\bm{\beta}_x,\label{eq:beta_x}\\
    \mathbf{a}(\zeta_{e,l}, r_l,N_y)&=\tilde{\mathbf{D}}_{y}\bm{\beta}_y,\label{eq:beta_y}
\end{align}

Specifically, $\tilde{\mathbf{D}}_x \in \mathbb{C}^{N_x \times N_x}$ and $\tilde{\mathbf{D}}_y \in \mathbb{C}^{N_y \times N_y}$ are the modified DFT matrices, respectively, and are defined as
\begin{align}
\tilde{\mathbf{D}}_{x}&= \mathrm{diag}(\bar{\mathbf{b}}(\zeta_a, r)) \, \mathbf{D}_x,
\label{eq:Dmux}\\
\tilde{\mathbf{D}}_{y} &= \mathrm{diag}(\bar{\mathbf{b}}(\zeta_e, r)) \, \mathbf{D}_y.
\label{eq:Dmuy}
\end{align}
$\bm{\beta}_x$ and $\bm{\beta}_y$ denote two block-sparse vectors, defined as
\begin{align}
\bm{\beta}_x=\tilde{\mathbf{D}}_{x}^{H}\mathbf{a}(\zeta_{a,l}, r_l,N_x)\label{beta_x},\\
\bm{\beta}_y=\tilde{\mathbf{D}}_{y}^{H}\mathbf{a}(\zeta_{e,l}, r_l,N_y)\label{beta_y}.
\end{align}
Substituting~\eqref{eq:beta_x} and~\eqref{eq:beta_y} into~\eqref{eq:H_final}, we arrive at
\begin{align}
    \label{H_Sigma_l}
    \mathbf{H}= \sum_{l=0}^{L-1}\tilde{g}_l\tilde{\mathbf{D}}_{x}\bm{\beta}_x \, (\tilde{\mathbf{D}}_{y}\bm{\beta}_y)^{T} =\tilde{\mathbf{D}}_{x}  \mathbf{\Sigma} \tilde{\mathbf{D}}_{y} ^{T},
\end{align}
where $\mathbf{\Sigma}\triangleq\sum_{l=0}^{L-1}\mathbf{\Sigma}_l $ and
\begin{align}
	\label{Sigmal_l}
   \mathbf{\Sigma}_l\triangleq \tilde{g}_l \, \bm{\beta}_{x,l} \, \bm{\beta}_{y,l}^T.
\end{align}

\begin{theorem}\label{thm:block-cartesian}
	 The matrix $\mathbf{\Sigma}_l$ as an outer product of $\bm{\beta}_{x,l}$ and $\bm{\beta}_{y,l}$ in \eqref{Sigmal_l} would admit a 2D block-sparse structure.
\end{theorem}
\begin{proof}
Let the block partition of $\bm{\beta}_{x,l}$ and $\bm{\beta}_{y,l}$ be
$\mathcal{B}x = \{K_{x,1}, \ldots, K_{x,S_x}\}$, $\mathcal{B}y = \{K_{y,1}, \ldots, K_{y,S_y}\}$, respectively.
The nonzero block index set of $\bm{\beta}_{x,l}$ and $\bm{\beta}_{y,l}$ are $\mathcal{I}_X \subseteq \{1, \ldots, S_x\}$ and $\mathcal{I}_Y \subseteq \{1, \ldots, S_y\}$.
From \eqref{Sigmal_l}, we have
\begin{align}
	\left|(\mathbf{\Sigma}_l)_{i,j}\right|=\left|\tilde{g}_l\right|\left|(\bm{\beta}_{x,l})_i\right|\cdot\left|(\bm{\beta}_{y,l})_j\right|.
\end{align}
Because $\left|\tilde{g}_l\right|$ is nonzero, $\left|(\mathbf{\Sigma}_l)_{i,j}\right|$ is nonzero if and only if $\left|(\bm{\beta}_{x,l})_i\right| \neq 0$ and $\left|(\bm{\beta}_{j,l})_j\right| \neq 0$, which indicates that the support set of $\mathbf \Sigma_l $ is
\begin{align}
	\operatorname{supp}(\mathbf{\Sigma}_l) &= \operatorname{supp}(\bm{\beta}_{x,l}) \times  \operatorname{supp}(\bm{\beta}_{y,l}) \nonumber \\
	&= \bigcup_{p\in\mathcal{I}_X} \; \bigcup_{q\in\mathcal{I}_Y}\{  K_{x,p}\times K_{y,q} \},
	\label{sigma_supp}
\end{align}
where $\times$ in~\eqref{sigma_supp} denotes the Cartesian product for the set. 
Since the nonzero support of $\mathbf{\Sigma}_l$ is determined by $\operatorname{supp}(\bm{\beta}_{x,l}) \times \operatorname{supp}(\bm{\beta}_{y,l})$, $\mathbf{\Sigma}_l$ can be regarded as a matrix consisting of multiple two-dimensional rectangular nonzero blocks. Therefore, $\mathbf{\Sigma}_l$ exhibits a 2D block-sparse structure.
\end{proof}
To illustrate the sparse structure, we simulate a $64\times64$ UPA with $L=3$, consisting of one line-of-sight (LoS) and two non-line-of-sight (NLoS) paths. The sparsity of $\bm{\Sigma}$ is shown in Fig.~\ref{fig:sparsity}, which clearly exhibits a 2D block-sparse structure under the derived dictionaries $\tilde{\mathbf{D}}_{x}$ and $\tilde{\mathbf{D}}_{y}$.
\begin{remark}
Due to the sparse scattering characteristic of mmWave/THz signals, $\mathbf{\Sigma}$ is primarily determined by the LoS path and the corresponding $\mathbf{\Sigma}_0$. Therefore, $\mathbf{\Sigma}$ admits a similar 2D block-sparse structure.
\begin{figure}[h]
	\centering
	\includegraphics[width=0.80\linewidth]{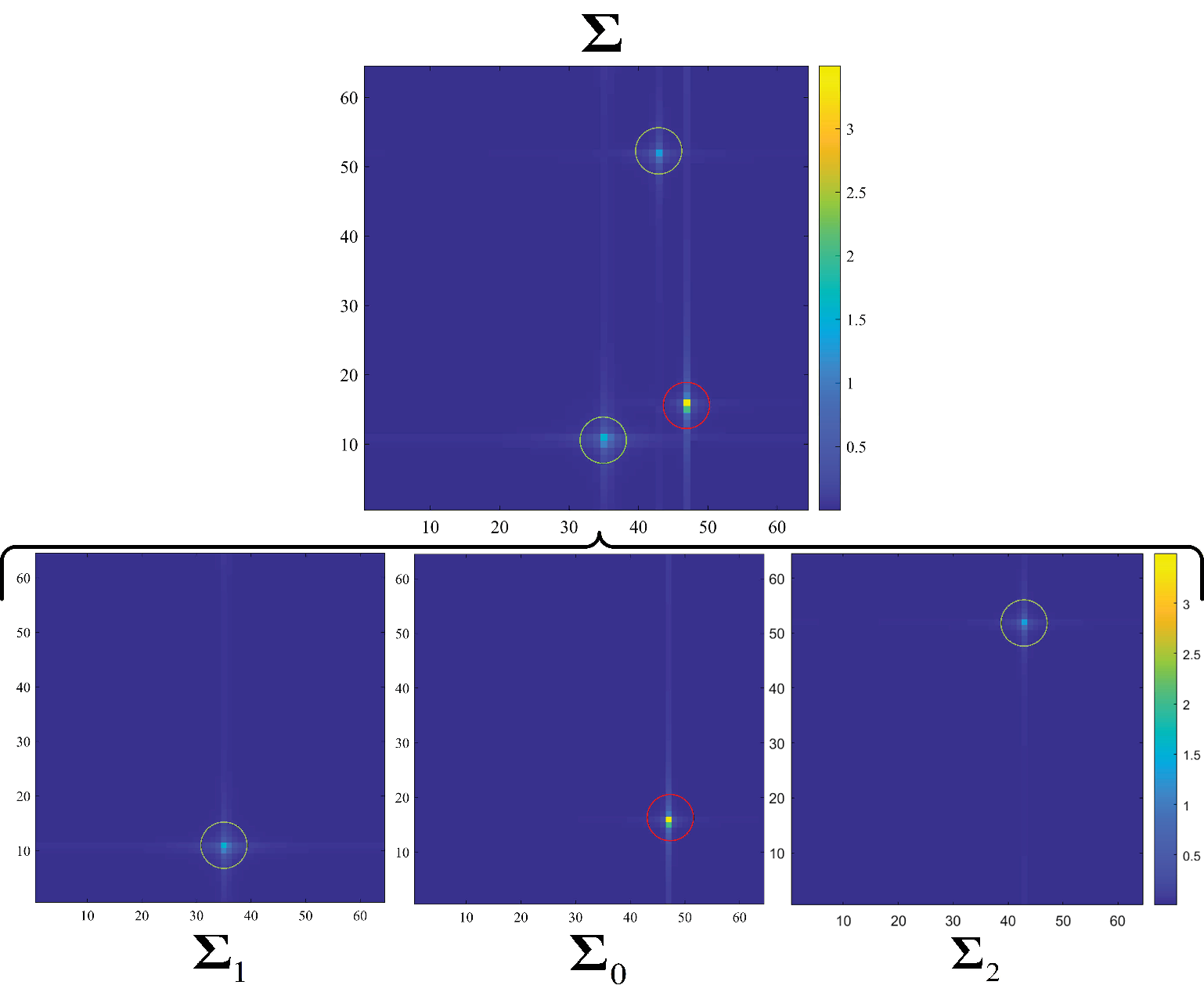}
	\caption{The 2D sparse structure of $\bm{\Sigma}$}
	\label{fig:sparsity}
\end{figure}

\end{remark}

\vspace{-2mm}
\subsection{2D Block-Sparse Aware Channel Estimations}
Based on \eqref{H_Sigma_l} the vectorized channel vector $\mathbf{h}$ is written as
\begin{align} \label{estimateh}
	\mathbf{h}=(\tilde{\mathbf{D}}_{y}\otimes\tilde{\mathbf{D}}_{x})\text{vec}{(\mathbf{\Sigma})} =\tilde{\mathbf{D}}\bm{\beta}, 
\end{align}
where $\bm{\beta} \triangleq \text{vec}{(\mathbf{\Sigma})}$ and 
\begin{align} \label{Dic}
	\tilde{\mathbf{D}} \triangleq \tilde{\mathbf{D}}_{y}\otimes\tilde{\mathbf{D}}_{x}.
\end{align} 
Substituting \eqref{estimateh} into \eqref{receivedmodelY}, the received signal is rewritten as 
\begin{align} \label{compresssensing}
	\mathbf{y} =\mathbf{F}\tilde{\mathbf{D}}\bm{\beta}+\mathbf{n}  =\mathbf{\Phi}\bm{\beta}+\mathbf{n},
\end{align}
where 
\begin{align} \label{Phi}
	\mathbf{\Phi} \triangleq \mathbf{F}\tilde{\mathbf{D}}.
\end{align} 
$\mathbf{\Phi}$ denotes the measurement matrix. Certainly, $\bm{\beta}$ is a 2D block-sparse vector derived from $\mathbf \Sigma$. The UPA channel estimation problem is now transformed into a 2D block-sparse compressive sensing problem.

It is worth noting that, based on the proposed novel dictionary $\tilde{\bm{D}}$ and the corresponding measurement matrix, the compress sensing problem in \eqref{compresssensing} can be solved by BOMP or PCSBL~\cite{Blocksparsesignal,1D}. However, neither BOMP or PCSBL can exploit the 2D block-sparse structure in $\bm{\beta}$. 

In this paper, we propose a 2D block-sparse aware channel estimation algorithm by applying the 2D-PCSBL method~\cite{twodimention}. Specifically, we aim to estimate the posterior distribution $p(\bm{\beta}|\mathbf{y},\bm{\alpha})$ and maximize it with respect to $\bm{\beta}$. In the posterior distribution $p(\bm{\beta}|\mathbf{y},\bm{\alpha})$, the sparsity of $\beta_n \triangleq \bm{\beta}[n]$ is jointly determined by its hyperparameter $\alpha_n \triangleq \bm{\alpha}[n]$ and its neighbor's hyperparameters. Details on the 2D-PCSBL algorithm are as follows.

As the observation $\mathbf{y}$ is conditionally independent of $\bm{\alpha}$ given $\bm{\beta}$, the posterior distribution $p(\bm{\beta}|\mathbf{y},\bm{\alpha})$ is given by
\begin{align}
	\label{beysian}
	p(\bm{\beta}|\mathbf{y},\bm{\alpha}) 
	\propto p(\mathbf{y}|\bm{\beta},\bm{\alpha})\,p(\bm{\beta}|\bm{\alpha})=p(\mathbf{y}|\bm{\beta})\,p(\bm{\beta}|\bm{\alpha})
	,
\end{align}
where $p(\mathbf{y}|\bm{\beta})$ and $p(\bm{\beta}|\bm{\alpha})$ denote the likelihood and prior distributions, respectively.

Since $\bm{\beta}$ is a 2D block-sparse vector derived from $\mathbf{\Sigma}$,
the prior $p(\bm{\beta}|\bm{\alpha})$ is modeled as a zero-mean complex Gaussian distribution with its variance coupled to the hyperparameter $\bm{\alpha}$ in two dimensions, given as
\begin{align} \label{betaprior}
	p(\bm{\beta}|\bm{\alpha}) 
	= \prod_{n=1}^{N} \mathcal{CN}\!\left(\beta_n \mid 0, \eta_n^{-1}\right),
\end{align}
where
\begin{align}
	\eta_n &\triangleq \alpha_n + \rho \sum_{k \in \mathcal{G}_{(n)}} \alpha_k, \label{eta} \\ 
	p(\bm{\alpha})&=\prod_{n=1}^N\mathrm{Gamma}(\alpha_n|a,b).\label{gamma}
\end{align}
Here, $p(\bm{\alpha})$ follows a Gamma distribution as in \eqref{gamma} with hyperpriors $a>1$ and $b=10^{-6}$. $\mathcal{G}_{(n)}$ represents the indices of hyperparameters that neighbor $\bm{\alpha}_n$, i.e., the indices of hyperparameters adjacent to $\bm{\alpha}_{i,j}$ on the 2D grid illustrated in Fig.~\ref{fig:twoD}. The parameter $\rho \in [0,1]$ denotes the correlation coefficient between $\bm{\alpha}_n$ and $\bm{\alpha}_k$, $k \in \mathcal{G}_{(n)}$.
The symbol $\lceil * \rceil$ denotes the ceiling operator.

\begin{figure}[h]
	\centering
	\includegraphics[width=0.85\linewidth]{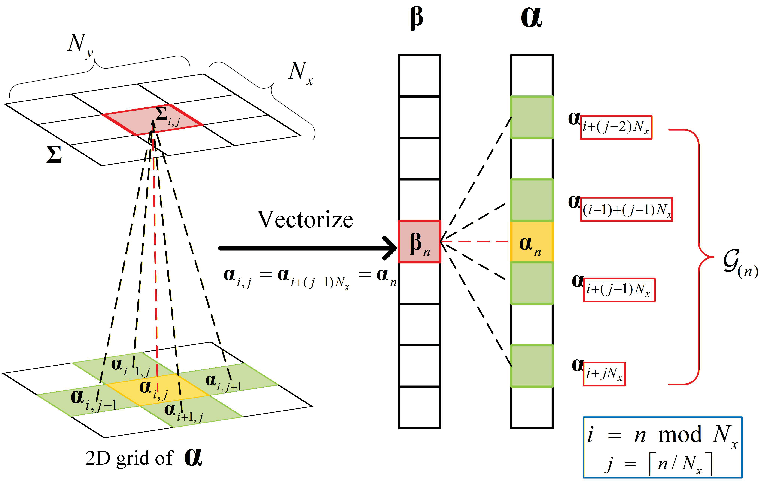}
	\caption{The 2D pattern-coupled structure of $\bm{\beta}$}
	\label{fig:twoD}
	\vspace{-4mm}
\end{figure}

Substituting~\eqref{betaprior} and~\eqref{gamma} into~\eqref{beysian}, the posterior $p(\bm{\beta}|\mathbf{y}, \bm{\alpha})$ can be readily verified to follow a complex Gaussian distribution, whose mean and covariance are given by~\cite{twodimention}
\begin{align}
	\bm{\mu}&=\gamma\bm{\Psi} \mathbf{\Phi}^H\mathbf{y}, \label{mu}\\
	\bm{\Psi}&=(\gamma\mathbf{\Phi}^H\mathbf{\Phi}+\bm{\Lambda})^{-1}, \label{Psi} 
\end{align}
where $\bm{\Lambda} \triangleq \text{diag}(\eta^{(t)}_1,\dots,\eta^{(t)}_N)$ and $ \gamma = 1/\sigma_n^{2} $ is the inverse of the noise variance in \eqref{compresssensing} which is assumed to be a known \textit{priori}. 

Given the estimated hyperparameters $\bm{\alpha}$, the maximum a posterior (MAP) estimate of $\bm{\beta}$ is the mean in~\eqref{mu}, given by
\begin{align}
	\hat{\bm{\beta}}=\boldsymbol{\mu}=(\mathbf{\Phi}^{H}\mathbf{\Phi}+\gamma^{-1}\boldsymbol{\Lambda})^{-1}\mathbf{\Phi}^{H}\boldsymbol{y}.
\end{align} 

Here, to obtain the $\hat{\bm{\beta}}$ we only need to estimate the hyperparameter $\bm{\alpha}$.
 So, an expectation–maximization (EM) algorithm is employed to estimate the hyperparameter $\bm{\alpha}$ through alternations between the E- and M-steps.

\textbf{1) E-Step:}
The posterior $p(\bm{\beta} \mid \mathbf{y}, \bm{\alpha}^{(t-1)})$ requires computing $\bm{\mu}^{(t)}$ and $\bm{\Psi}^{(t)}$ as in~\eqref{mu}–\eqref{Psi}. To avoid the costly matrix inversion in~\eqref{Psi}, the Generalized Approximate Message Passing (GAMP) algorithm~\cite{rangan2011generalized} is utilized to obtain approximate $\tilde{\bm{\mu}}^{(t)}$ and $\tilde{\bm{\Psi}}^{(t)}$ given $\bm{\alpha}^{(t-1)}$ ($t = 1, \ldots, t_s$).

\textbf{2) M-Step:} 
Compute $\alpha_n^{(t)}$ by maximizing the lower bound of the posterior probability~$p(\bm{\alpha}|\bm{y})$ also referred to as the Q-function based on the computed $p(\bm{\beta}|\mathbf{y},\bm{\alpha}^{(t-1)})$. Since $\bm{\alpha}$ only affects $\bm{\beta}$, where $p(\bm{\alpha}|\bm{y}) = p(\bm{\alpha}|\bm{\beta})$, the maximization of the Q-function can be expressed as
\begin{align}\label{Mstep}
	\bm{\alpha}^{(t)}&=\arg\max_{\bm{\alpha}} Q(\bm{\alpha}|\bm{\alpha}^{(t-1)}) \nonumber \\
	&= \arg\max_{\bm{\alpha}}\mathbb{E}_{\bm{\beta}|\mathbf{y},\bm{\alpha}^{(t-1)}}\left[\log p(\bm{\alpha}|\bm{\beta})\right] \nonumber \\
	&= \arg\max_{\bm{\alpha}}\mathbb{E}_{\bm{\beta}|\mathbf{y},\bm{\alpha}^{(t-1)}}\left[\log p(\bm{\beta}|\bm{\alpha})p(\bm{\alpha})\right]+c,
\end{align}
where the operator $\mathbb{E}_{\bm{\beta}|\mathbf{y}, \bm{\alpha}^{(t-1)}}[\cdot]$ denotes the expectation with respect to the posterior distribution $p(\bm{\beta}|\mathbf{y}, \bm{\alpha}^{(t-1)})$.
The maximization in \eqref{Mstep} is further rewritten as
\begin{align} \label{upgradealph}
	\alpha_n^{(t)}=\frac{a-1}{0.5\omega_n+b}\quad\forall n,
\end{align}
where
\begin{align}
	\label{omega}
	\omega_n \triangleq 
	\left[ \left(\tilde{\mu}_{n}^{(t)}\right)^{2} + \tilde{\Psi}_{nn}^{(t)} \right]
	+ \rho \sum_{k\in \mathcal{G}_{(n)}} 
	\left[ \left(\tilde{\mu}_{k}^{(t)}\right)^{2} + \tilde{\Psi}_{kk}^{(t)} \right].
\end{align}
with $\tilde{\mu}_n\triangleq \tilde{\bm{\mu}}[n]$ and $\tilde{\Psi}_{nn} \triangleq \tilde{\bm{\Psi}}(n,n)$, 

Finally, when the EM algorithm converges, we obtain $ \hat{\bm{\beta}} = \tilde{\bm{\mu}}^{(t_s)} $.
The proposed 2D-PCSBL based algorithm is summarized in Algorithm~\ref{alg1}.

\addtolength{\topmargin}{0.04in}
\begin{algorithm}[ht]
	\caption{Proposed 2D-PCSBL based Algorithm}
	\label{alg1}
	\KwIn{$\mathbf{y}$, $\mathbf{F}$, $\zeta_a$, $\zeta_e$, threshold $\epsilon$, max-iter $t_{\max}$}
	\KwOut{UPA near-field channel $\mathbf{H}$.}

	\textbf{Dictionary construction}: The modified dictionary $\tilde{\mathbf{D}}_{x}\leftarrow$ \eqref{eq:Dmux}, $\tilde{\mathbf{D}}_{y}\leftarrow$ \eqref{eq:Dmuy}, $\tilde{\mathbf{D}}\leftarrow$ \eqref{Dic}; the measurement matrix  $\bm{\Phi} \leftarrow$ \eqref{Phi}.\\
	\textbf{Initialization}: Initialize $\bm{\alpha}^{(0)}$, $\bm{\beta}^{(0)}$ and set $t\!\leftarrow\!0$.\\
	
	\While{$t<t_{\max}$ and $\|\hat{\bm{\beta}}^t-\hat{\bm{\beta}}^{t-1}\|^2_2 <=\epsilon $}{
	\textbf{E-step}: Estimate the approximate posterior mean $\tilde{\bm{\mu}}^{(t)}$ and covariance $\tilde{\bm{\Psi}}^{(t)}$ via the GAMP algorithm using $\bm{\alpha}^{(t-1)}$.\\
		\textbf{M-step}: Compute the hyperparameter $\bm{\alpha}^{(t)} \leftarrow$~\eqref{upgradealph}~\eqref{omega} using $\tilde{\bm{\mu}}^{(t)}$ and $\tilde{\bm{\Psi}}^{(t)}$.\\
		\textbf{State update}: $t \leftarrow t+1$.
	}
	\textbf{Reconstruction}: Let $\hat{\bm{\beta}} \leftarrow \tilde{\bm{\mu}}^{(t_s)}$, compute $\hat{\mathbf{h}}=\tilde{\mathbf{D}}\hat{\bm{\beta}}$ and $\hat{\mathbf{H}}=\mathrm{reshape}(\hat{\mathbf{h}},N_x,N_y)$.
\end{algorithm}
\vspace{-2mm}
\subsection{Complexity Analysis} 
The computational complexity analysis, encompassing both theoretical evaluation and simulation results, is summarized in Table~\ref{Complexitysimulation}, where $K_p$ and $K_o$ denote the numbers of iterations for the PCSBL algorithm and the OMP algorithm, respectively, and $N_o$ represents the number of atoms of the polar-domain representation.


By leveraging the GAMP framework, the proposed 2D-PCSBL algorithm circumvents the costly matrix inversion in~\eqref{Psi} and thus reduces the overall complexity from $\mathcal{O}(K_p N^3)$~\cite{tipping2001sparse} to $\mathcal{O}(K_p T N)$ compared with the OMP-based algorithm~\cite{vila2013expectation}.
\begin{remark}
It is worthwhile to note that the proposed 2D-PCSBL based algorithm exhibits almost the same computational complexity as the PCSBL algorithm.  
The major difference lies in that, whereas the PCSBL algorithm exploits only the block-sparse structure, the 2D-PCSBL algorithm fully captures the coupling of sparse structures across the horizontal and vertical dimensions of UPA near-field channels.
And so, the overall framework remains unchanged, and the extra computational overhead induced by the iterative update of the 2D block-sparse hyperparameter $\bm{\alpha}$ is negligible.
\end{remark}	
\vspace{-1mm}
\section{Numerical Simulation} 	
In this section, numerical simulations are conducted to evaluate the performance of the proposed method. We compare the proposed 2D-PCSBL based algorithm with several state-of-the-art methods, including the near-field OMP (polar-OMP) algorithm~\cite{Muitipleaccess}, the BOMP algorithm and the PCSBL algorithm. The Polar-OMP algorithm is applied based on the polar-domain dictionary. The BOMP and PCSBL algorithms are based on the derived dictionary $\tilde{\bm{D}}$ in~\eqref{Dic}, and try to exploit the channel's block-sparse structure. Further, the proposed 2D-PCSBL based algorithm aims to fully exploit the 2D block-sparse structure across the horizontal and vertical dimensions of the UPA near-field channels.

Simulation parameters are shown in Table~\ref{tab:simulation_parameters}, where $F_r=0.5\sqrt{D^3/\lambda}$ and $F_R=2D^2/\lambda$ represent Fresnel distance and Rayleigh distance, respectively~\cite{balanis2016antenna}. Here $D$ represents the aperture of antenna. Thus, in the simulation setup, the path distance $r_l$ lies within the entire near-field region. The signal-to-noise ratio (SNR) is defined based on \eqref{receivedmodelY}, i.e.,
\begin{align}
	\text{SNR(dB)}=10\log_{10}(\mathbb{E}\{\| \mathbf{Fh} \|_2^2\} /\mathbb{E}\{\| \mathbf{n} \|_2^2\}).
\end{align}
The performance is evaluated in terms of the normalized mean square error (NMSE), defined as
\begin{align} \label{NMSE}
	\text{NMSE(dB)}=10\log_{10}(\mathbb{E} \{\|\mathbf{H}-\hat{\mathbf{H}}\|^{2}_{\mathbb{F}}/ \|\mathbf{H}\|^2_{\mathbb{F}}\}),
\end{align}
where $\hat{\mathbf{H}}$ is the estimate of $\mathbf{H}$.
\begin{table}[t]
	\centering
	\caption{Simulation Parameters}
	\label{tab:simulation_parameters}
	\renewcommand{\arraystretch}{1.2}
	\vspace{-2mm}
	\begin{tabular}{|l|l|}
		\hline
		\textbf{Parameter} & \textbf{Value / Description} \\ \hline
		Carrier frequency $f_c$ & $0.1~\mathrm{THz}$ \\ \hline
		Bandwidth $B_w$ & $400~\mathrm{MHz}$ \\ \hline
		Array size $(N_x,N_y)$ & $(32,\,256)$ \\ \hline
		Sampling points $T$ & 1024 \\ \hline
		Precoding matrix $\mathbf{F}$ & i.i.d. $\mathcal{CN}(0,1/\sqrt{N})$ \\ \hline
		Number of paths $L$ & $3$ (1 LoS + 2 NLoS) \\ \hline
		NLoS path gain & $13~\mathrm{dB}$ lower than LoS \\ \hline
		AoA $\theta_l$ & $U[-\pi/3,\,\pi/3]$ \\ \hline
		AoD $\phi_l$ & $U[-\pi/3,\,\pi/3]$ \\ \hline
		Path distance $r_l$ & $U[F_r,F_R]$, with $F_r=2.73$, $F_R=99.84$ \\ \hline
		Block size (B) & 6 \\ \hline
	\end{tabular}
	\vspace{-4mm}
\end{table}

 Fig.~\ref{fig:SNRCompare} plots the NMSE performance with respect to SNR. One can see that the NMSE performance of all the algorithms based on the derived reconstructed dictionary $\tilde{\bm{D}}$ improves as the SNR increases. The NMSE performance of the Polar-OMP algorithm first decreases and then reaches a plateau, since it relies on the polar-domain dictionary, which ignores the block sparsity of $\bm{\beta}$ exploited by the derived dictionary $\tilde{\bm{D}}$. Moreover, the proposed 2D-PCSBL based algorithm achieves an NMSE of $-11$ dB at an SNR of $5$ dB, whereas the BOMP algorithm requires at least $10$ dB to attain the same NMSE performance. Furthermore the proposed algorithm exhibits superior robustness in low-SNR scenarios as compared with PCSBL. These improvements stem from the exploitation of the 2D block-sparse structure in UPA near-field channels.

Fig.~\ref{fig:TCompare} illustrates the NMSE performance versus the number of sampling points $T$ at an SNR of $5$ dB. As shown, all algorithms exhibit decreasing NMSE as $T$ increases. Notably, the proposed algorithm achieves an NMSE of $-8.9$ dB with only $T = 512$, significantly outperforming the polar-OMP algorithm and surpassing both BOMP and PCSBL. This demonstrates that exploiting the 2D block-sparsity is the key to reducing the pilots required for channel estimations.


Table~\ref{Complexitysimulation} presents a comparison of the computational complexity for different numbers of sampling points $T$ at $\text{SNR}=5$~dB. The complexity is evaluated in terms of the relative runtime (RT), defined as the ratio of the actual execution time to that of the reference implementation—specifically, the runtime of the polar-OMP algorithm with T = 512 sampling points. One can see that, the proposed 2D-PCSBL based method exhibits comparable computational complexity relative to existing approaches. Moreover, as $T$ increases, the computational complexity gap between the proposed method and the OMP-based algorithms narrows, highlighting the advantage of the proposed method in channel estimation scenarios with high pilot requirements.

\begin{table*}[t]
	\centering
	\fontsize{8.5}{10.5}\selectfont
	\caption{theoretical complexity and Relative runtime (RT) for UPA-ELAA channel estimation.}
	\label{Complexitysimulation}
	\setlength{\tabcolsep}{3pt}
	\renewcommand{\arraystretch}{0.95}
	\resizebox{\linewidth}{!}{%
		\begin{tabular}{|c|c|c|c|c|c|c|c|c|c|}
			\hline
			\multirow{2}{*}{Method} &
			\multirow{2}{*}{Dictionary Used} &
			\multirow{2}{*}{Main computational Steps} &
			\multicolumn{3}{|c|}{Theoretical Complexity} &
			\multicolumn{4}{|c|}{Relative Runtime (RT)} \\ 
			\cline{4-10}
			& & & Step 1 & Step 2 & Overall &
			$T=512$ & $T=768$ & $T=1024$ & $T=1280$ \\
			\hline
			Polar-OMP & Polar-domain & Correlation, LS
			& $\mathcal{O}(K_oTN_p)$ & $\mathcal{O}(TK_o^3)$ & $\mathcal{O}(K_oTN_p)$
			& \textbf{1.00} & 1.68 & 2.80 & \textbf{3.95} \\ \hline
			
			BOMP & Proposed dictionary & Correlation, LS (per block)
			& $\mathcal{O}(K_oTN/B)$ & $\mathcal{O}(TK_o^3/B)$ & $\mathcal{O}(K_oTN/B)$
			& \textbf{0.18} & 0.35 & 0.45 & \textbf{0.64} \\ \hline
			
			PCSBL & Proposed dictionary & E-step, M-step
			& $\mathcal{O}(K_pTN)$ & $\mathcal{O}(K_pN)$ & $\mathcal{O}(K_pTN)$
			& \textbf{2.65} & 3.77 & 4.34 & \textbf{5.06} \\ \hline
			
			Proposed Method & Proposed dictionary & E-step, M-step (2D)
			& $\mathcal{O}(K_pTN)$ & $\mathcal{O}(K_pN)$ & $\mathcal{O}(K_pTN)$
			& \textbf{2.72} & 3.84 & 4.41 & \textbf{5.14} \\ \hline
		\end{tabular}%
	}
	\vspace{-2mm}
\end{table*}

\begin{figure}[t]
	\vspace{-3mm}
    \centering
    \includegraphics[width=0.82\linewidth]{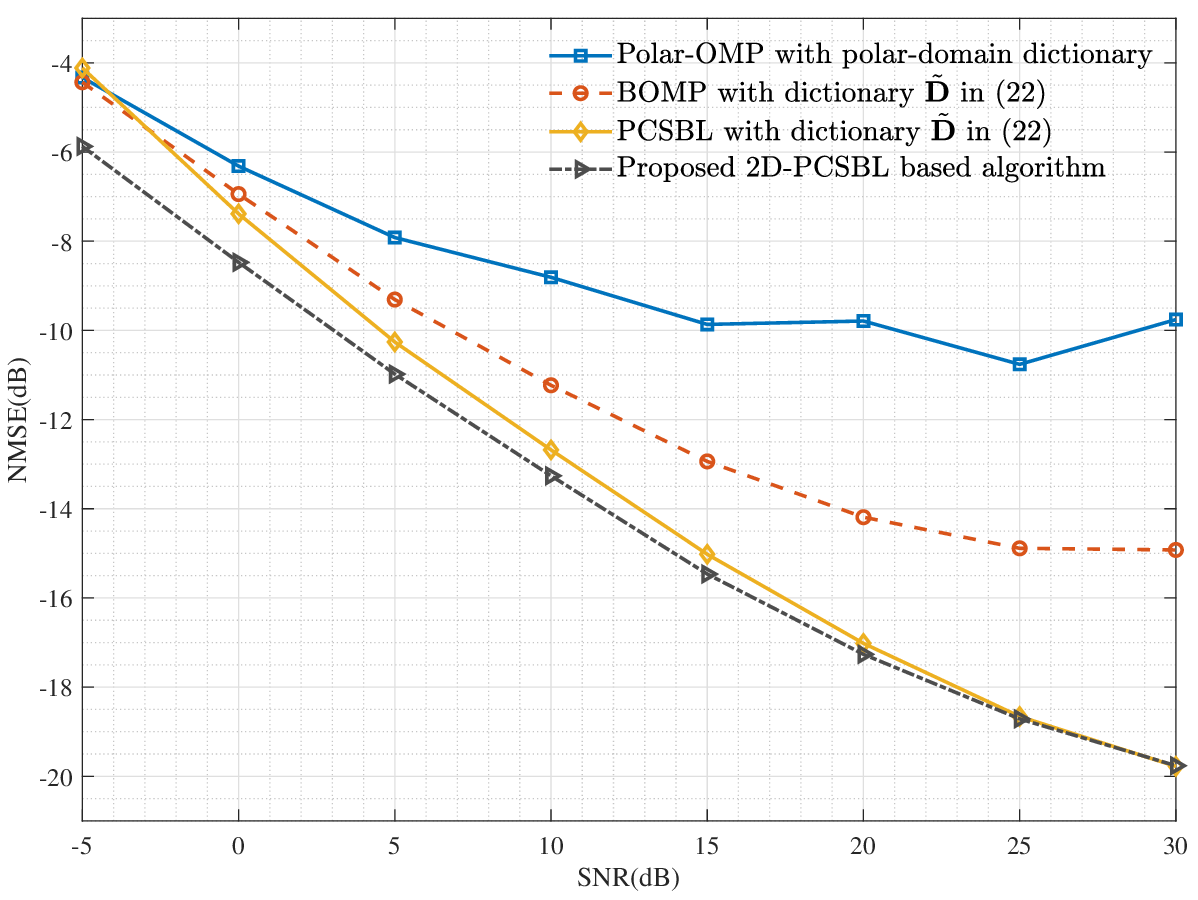}
    \caption{NMSEs of respective algorithms vs. SNR}
    \label{fig:SNRCompare}
\end{figure}

\begin{figure}[t]
	\vspace{-3mm}
    \centering
    \includegraphics[width=0.82\linewidth]{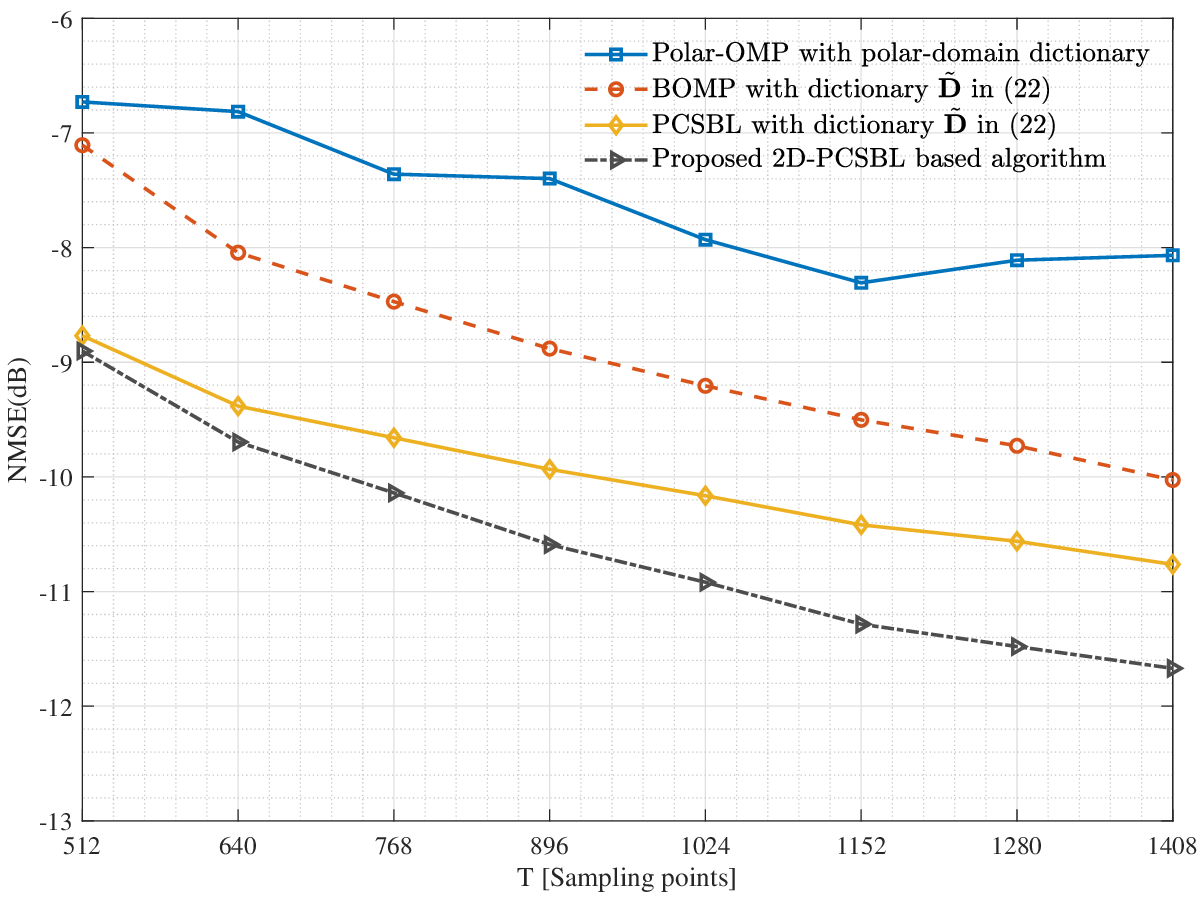}
    \caption{NMSEs of respective algorithms vs. T}
    \label{fig:TCompare}
    \vspace{-3mm}
\end{figure}
\vspace{-2mm}
\section{Conclusions}
 In this work, we have exploited the coupling feature of the sparsity across the horizontal and vertical dimensions of the UPA near-field channels to improve channel estimation accuracy with fewer pilots. Specifically, we first reformulated the UPA near-field channel as an outer product of two ULA near-field channels and designed a modified 2D-DFT matrix as the sparse dictionary for it. Under this dictionary, we have reformulated the original channel estimation problem as a 2D-block sparse recovery problem, and have efficiently solved it by the 2D-PCSBL algorithm. Simulation results show that the aforementioned 2D-block sparse coupling feature is key in improving estimation accuracy and reducing pilots as compared with existing methods.

In future work, we will extend the proposed framework to broadband near-field channels by formulating a three-dimensional block-sparse model across subcarriers, thereby enabling efficient wideband channel estimation. 
\vspace{-1mm}

\bibliographystyle{IEEEtran}
\bibliography{ICC}

\end{document}